\documentclass[a4paper,12pt]{article}
\usepackage{amsmath,amssymb,amsthm,url}
\usepackage{epsfig,graphicx}
\usepackage{hyperref}
\usepackage[margin=3cm]{geometry}
\usepackage[ruled,noend,noline,linesnumbered]{algorithm2e}
\usepackage[titletoc]{appendix}
\usepackage{authblk}
\usepackage{mathtools}

\newtheorem{lemma}{Lemma}[section]
\newtheorem{proposition}[lemma]{Proposition}
\newtheorem{corollary}[lemma]{Corollary}
\newtheorem{theorem}[lemma]{Theorem}

\newtheorem{conjecture}[lemma]{Conjecture}

\theoremstyle{definition}
\newtheorem{remark}[lemma]{Remark}

\usepackage{booktabs,caption,fixltx2e}
\usepackage[flushleft]{threeparttable}

\let\oldnl\nl
\newcommand{\nonl}{\renewcommand{\nl}{\let\nl\oldnl}}

\def\lc{\left\lceil}   
\def\rc{\right\rceil}

\title{%
  The simultaneous conjugacy problem \\ in the symmetric group
}

\author[a,b]{Andrej Brodnik\thanks{%
    This work is sponsored in part by the Slovenian Research Agency (research program P2-0359 and research projects J1-2481, J2-2504, and N2-0171).
  }
}
\author[a,c,d]{Aleksander Malni\v{c}\thanks{%
    This work is sponsored in part by the Slovenian
    Research Agency
    (research program P1-0285 and
    research projects N1-0062, J1-9108, J1-9110, J1-9187, J1-1694, J1-1695).
  }
}
\author[a]{Rok Po\v{z}ar\thanks{%
    Corresponding author. This work is sponsored in part by the Slovenian
    Research Agency
    (research program P1-0404 and
    research projects N1-0062, J1-9110, J1-9187, J1-1694).
  }
}
\affil[a]{University of Primorska, IAM/FAMNIT}
\affil[b]{University of Ljubljana, FRI}
\affil[c]{University of Ljubljana, PEF}
\affil[d]{IMFM}

\begin{document}

\maketitle

\begin{abstract}
The transitive simultaneous conjugacy problem asks whether there exists a permutation $\tau \in S_n$ such that $b_j = \tau^{-1} a_j \tau$ holds for all $j = 1,2, \ldots, d$, where
$a_1, a_2, \ldots, a_d$ and $b_1, b_2, \ldots, b_d$ are given sequences of $d$ permutations in $S_n$, each of which
generates a transitive subgroup of $S_n$.
As from mid 70' it has been known that the problem can be solved in $O(dn^2)$ time. An algorithm with
running time $O(dn \log(dn))$, proposed in late 80', does not work correctly on all input data. In this paper we solve  the transitive simultaneous conjugacy problem in $O(n^2 \log d / \log n + dn\log n)$ time and $O(n^{3/ 2} + dn)$ space. Experimental evaluation on random instances shows that the expected running time of our algorithm is considerably better, perhaps even nearly linear in $n$ at given $d$.


\noindent
\textbf{MCS}: 05C85, 05C60.

\end{abstract}

\section{Introduction}

The {\em $d$-Simultaneous conjugacy problem} (or $d$-SCP for short)  in a group $G$ asks the following: given two
ordered $d$-tuples $(a_1, a_2,\ldots , a_d)$ and $(b_1, b_2,\ldots , b_d)$ of elements from $G$, 
does there exists an element $\tau$ in $G$ such that $b_j = \tau^{-1}a_j \tau$ holds for all indices 
$j = 1,2,\ldots, d$?

This  problem  arises naturally in various fields of mathematics, computer science, and their applications in numerous
 other fields such as biology, chemistry etc. -- most notably when
deciding  whether two objects from a given class of objects are isomorphic (which usually means structural
equivalence). We quickly review some easy examples from the theory of maps on surfaces, the theory of 
covering graphs, computational group theory,  representation theory, interconnection networks, and cryptography.
We refer the reader to relevant literature for more information on these topics \cite{Eick, braid, MNS00, MNS02, Seress, Srid89, Oruc}.

A finite graph embedded in a closed oriented surface is conveniently given in terms 
of two permutations $R$ and $L$ on the set of arcs in the graph, where $R$ represents local rotations of arcs around 
each vertex and $L$ is the arc-reversing involution. It is generally known that two embeddings of the same graph 
given by $(R,L)$ and $(R',L')$ are combinatorially equivalent if and only if there exists a permutation $\tau$ of the 
arc set such that $R' = \tau^{-1} R \tau$ and $L' = \tau^{-1} L \tau$. Thus, the problem translates to a $2$-SCP
 in the symmetric group $S_n$.

An $n$-fold covering projection of graphs can be given by assigning to each arc of the base graph a certain 
permutation from  the symmetric group $S_n$, the so-called voltage of the arc; voltage assignments to arcs
 naturally extend to all walks, and in particular, to all fundamental closed walks at a base vertex 
 which generate the fundamental group.  Now, two covering projections given by two voltage assignments are
  equivalent  (which roughly means that the covering graphs are the same modulo relabeling of vertices and arcs)
  if and only there is a permutation in $S_n$ that simultaneously conjugates the two  ordered tuples of 
  voltages assigned to fundamental closed walks at a given base vertex. Thus, the problem translates to 
  $d$-SCP in $S_n$, where $d$ is the Betti number of the graph.

In computational group theory,  an important variant  of  SCP in  $S_n$  occurs when 
the two tuples are equal, that is, $a_j = b_j$ for all $j = 1,2,\ldots, d$; in this case we ask whether there exists a non-trivial
element $\tau\in S_n$ that belongs to the centralizer of the permutation group generated by $a_1, a_2, \ldots, a_d$.
   Often we are interested  in  actual construction of  such permutations  $\tau$, if they exsist.
 For, this is encountered in computing  the generators of   automorphism groups of embedded graphs
 on surfaces.

 In representation theory, the problem of whether two representations of a given group 
    on a vector space are equivalent translates to the simultaneous conjugacy problem  in the general linear 
   group, where the two tuples are  matrices representing the generators of the group. 
   
 As for problems encountered in computer science,  the equivalence of two permutation networks 
  under the permutation of inputs    and outputs also translates to SCP in the symmetric group. Last but not least,     
     SCP  in  braid groups  has been used in  attacking
    cryptographic protocols.

The paper aims  to study $d$-SCP in the symmetric  group $S_n$.
For $d=1$ the problem translates  to the question of whether two permutations have the same cycle type;  this can be tested
  in $O(n)$ time. So, the problem is interesting if $d> 1$. The condition that, for each $j=1,2, \ldots, d$, the
   permutations  $a_j$ and $b_j$ have the same cycle type is clearly necessary for the two tuples to be
     simultaneously conjugated. However, it is not sufficient.

\subsection{Related work and motivation}
Starting with a  purely group-theoretic context of finding centralizers,  in 1977,
Fontet  gave   the  $O(dn^2)$-time algorithm for solving  the $d$-SCP  in $S_n$
     by translating it into a  special case of a colored digraph isomorphism problem \cite{Fontet}. 
     More precisely, let $G_a$ be an arc-colored digraph  with the vertex set $ \{1, 2,\ldots , n\}$ where 
      there is an arc $(u,v)$ colored $j$ if and only if $a_j(u) = v$,  $j=1,2,\ldots, d$. (Note that such a digraph can have multiple parallel arcs and loops.)  The arc-colored
       digraph $G_b$ is defined similarly on the vertex set $ \{1, 2,\ldots , n\}$ with
        an arc $(u,v)$ colored $j$ if and only if $b_j(u) = v$,  $j=1,2,\ldots, d$.  The permutation $\tau$ 
        from $S_n$ that simultaneously conjugates the two tuples  is now precisely a color and direction
         preserving digraph isomorphism from  $G_a$ onto $G_b$ (assuming that permutations in $S_n$ are multiplied from left to right). Fontet's algorithm was independently discovered by Hoffmann in 1982 \cite{Hoff82}.

In studying permutation networks, Sridhar  exhibited an $O(dn\,\mathrm{log}(dn))$-time algorithm  
 in  1989 \cite{Srid89} for the case in which the corresponding digraphs are strongly connected,
  that is, when the two tuples  generate transitive subgroups in $S_n$. In his solution, he applied techniques
   used by  Hopcroft and Tarjan for testing isomorphism of  planar 3-connected graphs \cite{HT73}.
     The underlying idea of this approach is to compute the automorphism partition  of the disjoint union of
      the two digraphs, that is, the orbits of the corresponding automorphism group. An isomorphism between
       the two digraphs exists if and only if there is an orbit that
contains a vertex from each digraph. In the context of arc-colored digraphs $G_a$ and $G_b$ associated
 with the simultaneous conjugacy problem,
this amounts to finding the orbits of the color and direction preserving automorphism group of the disjoint 
union $G_a + G_b$.  The construction consists of the preprocessing stage and the partition refinement stage. The main problem to be resolved is to find an initial partition  such that its refinement  provides the automorphism partition.  The key observation is that if the automorphism partition is nontrivial, then we must start with a nontrivial initial partition,  otherwise the refinement process returns the trivial partition.

 Unfortunately, as we have recently 
 discovered, Sridhar's algorithm does not work correctly on all inputs,  for example, in certain cases when all cycles of a given color are of the same length. 
Sridhar's preprocessing stage  for the case when all cycles of a given color are of the same length is presented in algorithm {\sc  ArcLabeling} in  Appendix~\ref{append}. It labels each arc of a digraph with a pair of the form $\langle \alpha, \beta \rangle$, where $\alpha$ is one of 0, 1, and 2, and $\beta$ is a nonnegative integer.  Assigning $\alpha = 0$ to the arcs  coloured 1 (which defines reference cycles), let other arcs receive $\alpha = 1$ if they join vertices on the same  reference cycle, and $\alpha = 2$ otherwise; the parameter $\beta$ records  ``relative jumps''  with respect to reference cycles,  where arcs of the reference cycles receive $\beta = 0$. Two arcs of the same color are then in the same cell of the initial partition if and only if they have the same label. 

 We now give an example on which the algorithm  {\sc  ArcLabeling} returns the trivial  initial partition, although 
the automorphism partition is nontrivial. Let
\begin{align*}
a_1 &= (1,2,3)(4,5,6)(7,8,9)(10,11,12)\\
a_2 &= (1,11)(2,4)(5,7)(8,10)(3,9)(6,12)
\end{align*} 
be permutations in $S_{12}$, and consider the corresponding arc-colored digraph $G_a$, which is  shown in Figure~\ref{fig:permut}.
It is obvious that by  calling  {\sc  ArcLabeling} on $G_a$,  arcs of color 1 (solid line) receive label $\langle 0, 0\rangle$ while   arcs of color 2 (dashed line) receive label $\langle 2, 0\rangle$. Consequently, the initial partition on the set of arcs of $G_a$ is trivial, which in turn results in the trivial automorphism partition of $G_a$. However, 
 the automorphism partition of $G_a$ is  actually  $\{  \{1, 4, 7, 10 \}, \{ 2, 5, 8, 11\},\{ 3, 6, 9, 12\} \}$, which is  nontrivial.
 
\begin{figure}[htpb]
\begin{picture}(300,200)
\put(120,10){\includegraphics[scale=0.5]{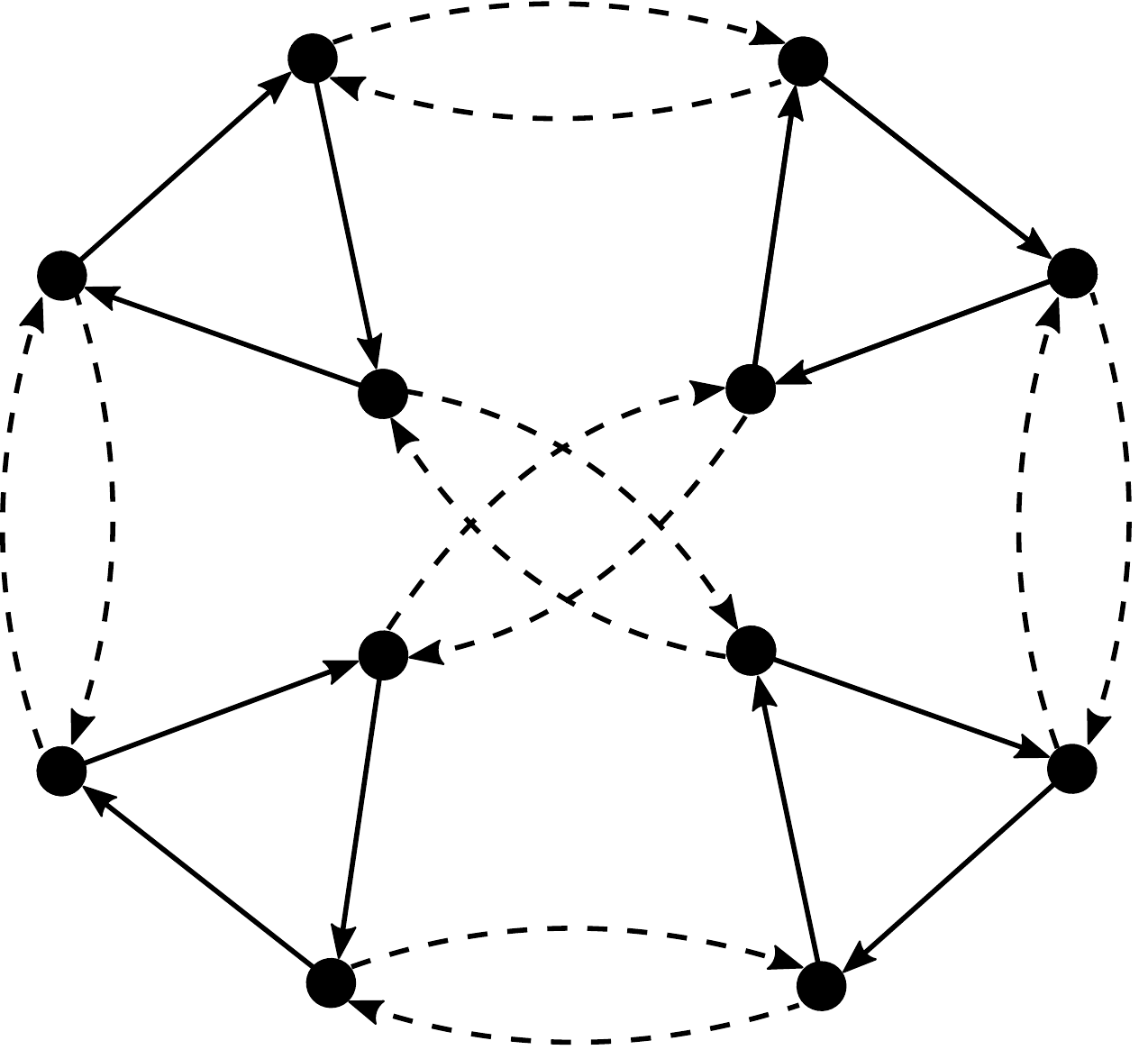}}
\put(155,173){$11$}
\put(295,140){$2$}
\put(231,123){$3$}
\put(231,60){$6$}
\put(185, 60){$9$}
\put(185, 123){$12$}
\put(115,140){$10$}
\put(121,42){$8$}
\put(294,42){$4$}
\put(253,173){$1$}
\put(165,8){$7$}
\put(255,8){$5$}
\end{picture}
\caption{An arc-colored digraph with a nontrivial automorphism partition.}\label{fig:permut}
\end{figure}
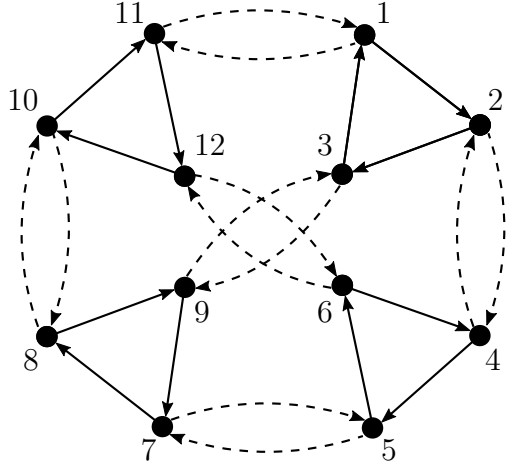

\subsection{Our results}
The question that arises naturally is, therefore, the following: 

\begin{quote}\it
Given two
ordered $d$-tuples $(a_1, a_2,\ldots , a_d)$ and $(b_1, b_2,\ldots , b_d)$ of elements from the symmetric group $S_n$, 
can we find  an element $\tau \in S_n$ such that $b_j = \tau^{-1}a_j \tau$ holds for all indices 
$j = 1,2,\ldots, d$ in time   $o(dn^2)$?
\end{quote}
The following main result answers the question affirmatively subject to the  condition that the two given tuples of permutations generate transitive permutation groups. We call this restricted problem the {\em transitive} $d$-SCP.
\begin{theorem}\label{thm:main}
The transitive $d$-SCP in the symmetric group $S_n$ can be solved   in $O(n^2  \log d /  \log n + dn\log n)$  time, using  $O(n^{3/2} + dn)$ space.
\end{theorem}

Our second result implies that for a large subclass of tuples generating transitive subgroups the problem can be solved in strongly subquadratic time in $n$ at a given $d$. More precisely, we prove the following.

\begin{theorem}\label{th:strongly}
Let  $(a_1, a_2,\ldots , a_d)$ and $(b_1, b_2,\ldots , b_d)$ be $d$-tuples  of permutations  that generate transitive subgroups of the symmetric group $S_n$, and let  $\lambda$ be the minimal number for which there are permutations $a_j$ and $b_j$ in the two tuples that have precisely $\lambda$  cycles in their cycle decompositions. Then we can test whether one tuple is simultaneously  conjugate to the other in $S_n$ in $O((d + \lambda) n  \log n)$ time, using $O((d + \log n)n)$  space.
\end{theorem}

The last theorem implies strongly subquadratic time in $n$  as soon as $\lambda = O(n^{\epsilon})$ for some constant $0 \leq \epsilon < 1$.
In particular, for $\lambda =1$ we give a linear time solution. 

\begin{theorem}\label{th:linear}
Let  $(a_1, a_2,\ldots , a_d)$ and $(b_1, b_2,\ldots , b_d)$ be $d$-tuples  of permutations so that, for some $j\in [n]$, the permutations $a_j$ and $b_j$ are $n$-cycles. Then we can test whether one tuple is simultaneously  conjugate in $S_n$ to  the other one in $O(d n)$ time and $O(d n)$ space.
\end{theorem}

The structure of the paper is the following. Section~\ref{sec:prelim} contains the necessary notation and basic definitions  to make the paper self-contained. In Section~\ref{sec:initial} we develop an algorithm which solves  the  transitive $d$-SCP in the symmetric group $S_n$ in $O(n^{2} + dn \log n)$ time and $O(dn)$ space.
In Sections~\ref{sec:shorter} and \ref{sec:storngly} we then give two improvements of  the basic algorithm from Section~\ref{sec:initial}  to prove  Theorems~\ref{thm:main} and \ref{th:strongly}, respectively. 
Theorem~\ref{th:linear} is proven in Section~\ref{linear} by  a completely different approach.
Section~\ref{sec:empir} presents some
empirical results. The last section  concludes the paper by discussing  some open problems related to our results.

\section{Preliminaries} \label{sec:prelim}

This section aims to establish some notation and terminology used in this paper. For the concepts not defined here see \cite{Diestel}.
The notations are quite graph centered since in our solution to the transitive simultaneous conjugacy problem we use graph-based techniques in the spirit of Fontet and Hoffmann.

For a positive integer $n$, we denote the  set $ \{1, 2,\ldots , n\}$ by $[n]$. For $i\in [n]$ and $g\in S_n$, we write $i^g$ for the image of $g$ under the permutation $g$ rather than by the more usual $g(i)$.
 Further, 
let $a = (a_1,a_2,\ldots, a_d)$ be a $d$-tuple of permutations in $S_n$. 
The {\bf permutation digraph} of $a$ is a pair $G_a = (V, A)$, where $V(G_a)= V = [n]$ is  the set of {\bf vertices}, and $A(G_a) = A$ is the set of ordered pairs $(i, a_k)$, $i \in [n], k\in [d]$,  called {\bf arcs}.
The {\bf degree} of $G_a$ is $|a|$. An arc $e = (i, a_k)$ has its {\bf initial vertex} $\textrm{ini}(e) = i$,   {\bf terminal vertex}  $\textrm{ter}(e) = i^{a_k}$, and {\bf color} $c(e) =k$; the arc $e = (i, a_k)$  is also referred  to as {\bf outgoing} from the vertex $i$. The vertices $\textrm{ini}(e)$ and  $\textrm{ter}(e)$ are {\bf end-vertices} of $e$. 


A {\bf  walk}  from a vertex $v_0$ to a vertex $v_m$ in a permutation digraph $G_a$ is an alternating   sequence  $W= v_0, e_1, v_1, e_2, \ldots, e_m, v_{m}$ of  vertices and arcs from $G$ such that for each  $i\in [m]$, the vertices $v_{i-1}$ and $v_{i}$ are the end-vertices of the arc $e_i$.  The vertices $v_0$ and $v_m$ are called the {\bf initial vertex} and the {\bf terminal vertex} of $W$, respectively. If  $v_0 = v_{m}$,  then the walk $W$ is  {\bf closed}, and it is {\bf open} otherwise.  A {\bf path} is a walk whose  vertices (and
hence arcs) are all pairwise distinct.

Walks defined so far are undirected. On the other hand we have also a  {\bf directed walk} in which case $v_{i-1} = \textrm{ini}(e_i)$ and  $v_{i} = \textrm{ter}(e_i)$ for each  $ i \in [m]$. A digraph  $G_a$  is {\bf weakly connected} if any vertex is reachable from any other vertex by traversing a  walk, and is {\bf strongly connected} if any vertex is reachable from any other vertex by traversing a directed walk.
 Clearly,  $G_a$  is strongly connected if and only if the tuple $a$ generates a transitive subgroup of $S_n$. 

A {\bf subdigraph} $H$ of $G_a$ consists of a subset $V(H) \subseteq V(G_a)$  and a subset $A(H) \subseteq A(G_a)$ such that every arc in  $A(H)$ has both end-vertices in $V(H)$.  A walk in a subdigraph $H$ of $G_a$ is a walk in $G_a$ consisting only of arcs from $A(H)$. 
A subdigraph $T$ of $G_a$ is a {\bf tree} if for any two vertices $v$ and $v'$ in $T$, there is a unique path from $v$ to $v'$ in $T$. If $V(T) = V(G_a)$, then $T$ is a {\bf spanning tree} of $G_a$.

A {\bf color-isomorphism} between two permutation digraphs $G_a$ and $G_b$  is a pair $(\phi_V, \phi_A )$ of bijections, where $\phi_V \colon V(G_a) \to V(G_b)$ and $\phi_A \colon A(G_a) \to A(G_b)$ such that $\phi_V(\textrm{ini}(e)) = \textrm{ini}(\phi_A(e))$, $\phi_V(\textrm{ter}(e)) = \textrm{ter}(\phi_A(e))$ and $c(e) = c(\phi_A(e))$ for any arc $e\in A(V_a)$.  

A {\bf partition} $\mathcal{P}$ of a set $V$ is a set $\mathcal{P} = \{V_1, V_2, \ldots,V_m\}$ of nonempty pairwise disjoint subsets of $V$ whose union is $V$. Each set $V_i$
is called a {\bf cell} of the partition. A partition is
{\bf trivial}  if it has only one cell, $\{ V\}$.

Let $\mathcal{G}$ be the set of all strongly connected  permutation digraphs on $n$ vertices, each of degree $d$,  and let $\mathcal{G}_V = \{(G_a,v) \,:\, G_a\in \mathcal{G}, v \in V(G_a)\}$ be the set of all rooted digraphs from $\mathcal{G}$. Given a set  $\mathcal{L}$ of labels, a {\bf vertex invariant} is a function  $I \colon \mathcal{G}_V \to \mathcal{L}$ 
such that whenever there exists a color-isomorphism  $G_a \to G_b$ mapping $v$ onto $w$, then
$$I(G_a,v) = I(G_b, w).$$
Since a vertex invariant $I$ assigns a label to every vertex of a permutation digraph $G_a$, it  induces a partition $\mathcal{P}_{I}(G_a) = \{V_1(G_a), V_2(G_a),\ldots,V_m(G_a)\}$ on the set $V(G_a)$ such that for any two vertices $u,v\in V(G_a)$, $u$ and $v$ are in the same cell if and only if they have the same labels $I(G_a,u) = I(G_a, v)$.   Two partitions $\mathcal{P}_{I}(G_a)$ and $\mathcal{P}_{I}(G_b)$  are {\bf compatible}  if $|\mathcal{P}_{I}(G_a)| = |\mathcal{P}_{I}(G_b)| = \mu$, and the cells of $\mathcal{P}_{I}(G_b)$ can be re-enumerated such that for  all $i\in[\mu]$ we have that $|V_i(G_a)| = |V_i(G_b)|$ and $I(G_a,v) = I(G_b, w)$ for all $v\in V_i(G_a)$ and $w\in V_i(G_b)$. 
The following results follow directly from  definitions.
\begin{lemma}\label{compatible}
Let $I \colon \mathcal{G}_V \to \mathcal{L}$ be a vertex invariant, and let $G_a, G_b\in \mathcal{G}$. Then:
\begin{enumerate}
\item[(i)] If $G_a$ and $G_b$ are color-isomorphic, then the partitions  $\mathcal{P}_{I}(G_a)$ and $\mathcal{P}_{I}(G_b)$  are compatible.
\item[(ii)]  Let $\mathcal{P}_{I}(G_a)$ and $\mathcal{P}_{I}(G_b)$ be compatible partitions, and let $v\in V_1(G_a)$. Then $G_a$ and $G_b$ are not color-isomorphic if and only if  there is no color-isomorphism of $G_a$ onto $G_b$ mapping $v$ to some vertex in $V_1(G_b)$. \end{enumerate}
\end{lemma}

An {\bf alphabet} $\Sigma$ is a finite set of  {\bf characters}. 
A {\bf word} (or a {\bf string}) over the alphabet $\Sigma$ is a finite sequence of characters from $\Sigma$.  
The {\bf length} $|\kappa|$ of a word $\kappa$ is the number of characters in  $\kappa$. 
The unique word $\epsilon$ of length 0 is the {\bf empty} word. We denote the set of all words of length $m$ over the alphabet $\Sigma$ by $\Sigma^m$.  We write $\kappa\kappa'$ for the concatenation of words $\kappa$ and $\kappa'$.
For a word $\kappa$ and a non-negative integer $m$  we define ${\kappa}^m = \epsilon$ for $m=0$, and otherwise ${\kappa}^m = {\kappa}^{m-1} \kappa$. 
Two words $\kappa$ and $\kappa'$ are {\bf cyclically equivalent} if there exist words $x$ and $y$ so that $\kappa = xy$ and $\kappa' = yx$.

 Let $G_a$ be a permutation digraph of degree  $d$ and let $\Sigma = \{k^{\alpha} \,|\, k \in [d], \alpha = \pm 1\}$ be an alphabet. 
A walk $W= v_0, e_1, v_1, e_2, \ldots, e_m, v_{m}$  in $G_a$   {\bf defines}   the word $\kappa(W) = k_1^{\alpha_1}k_2^{\alpha_2}\cdots k_m^{\alpha_m}$ over $\Sigma$, where: (i) $k_i = c(e_i)$; and (ii), if $\textrm{ini}(e_i) = v_{i-1}$ and $\textrm{ter}(e_i) = v_{i}$, then  $\alpha_i = +1$, and $\alpha_i = -1$ otherwise.

 Conversely, a word $\kappa'$ over $\Sigma$ and a vertex  $v_0 \in V(G_a)$ together {\bf define} the walk $W(\kappa', v_0)$ in $G_a$ starting at $v_0$ such that $\kappa(W(\kappa', v_0)) = \kappa'$. Finally, for a word $\kappa = k_1^{\alpha_1}k_2^{\alpha_2}\cdots k_m^{\alpha_m}$ over $\Sigma$ we define the permutation $a_{\kappa} = a_{k_1}^{\alpha_1}a_{k_2}^{\alpha_2}\cdots a_{k_m}^{\alpha_m}$.

Let  $ \{{\tt C}, {\tt O}\}$ be a set of labels and let $\kappa$ be a word over $\Sigma$. The function $I^{\kappa} \colon \mathcal{G} \to \{{\tt C}, {\tt O}\}$ defined by 
$$
  I^{\kappa}(G_a,v)=\left\{
  \begin{array}{@{}ll@{}}
    {\tt C}, & \text{if}\   \text{the walk $W(\kappa, v)$ in $G_a$ is closed}\\
    {\tt O}, & \text{otherwise}
  \end{array}\right.
$$ 
is vertex invariant.   In group-theoretical terms, vertices labeled by {\tt C} represent the set of fixed points  of the permutation $a_{\kappa}$, while vertices labeled by {\tt O} represent its  support set.  For a set $B \subseteq V(G)$ we let
 $$O_{\kappa}(B) = \{v \in B\,| \, I^{\kappa}(G,v)= {\tt O}\} \, \text{ and }\, C_\kappa(B) = \{v \in B \,| \, I^{\kappa}(G,v)= {\tt C}\}.$$
 For $\kappa=\epsilon$, we have $I^{\epsilon}(G,v) = {\tt C}$ for all $(G,v) \in \mathcal{G}$, and therefore $C_{\epsilon}(V(G)) = V(G)$ and  $O_{\epsilon}(V(G)) = \emptyset$.

Let $G_a$ and $G_b$ be two permutation digraphs, and 
let $W_1$ and $W_2$ be walks in $G_a$ and $G_b$, respectively. We say that $W_2$ is the walk {\bf corresponding} to the walk $W_1$ if $\kappa(W_1) = \kappa(W_2)$. In particular,  an arc $e = (v, a_k) \in A(G_a)$  {\bf corresponds} to an arc $f= (w, b_{k'})\in A(G_b)$  if $k=k'$.
Further, a word $\kappa$ over $\Sigma$  is a {\bf distinguishing word} for vertices $v \in V(G_a)$ and $w\in V(G_b)$  if $I^{\kappa}(G_a,v) \neq I^{\kappa}(G_b,w)$. In this case,  $v$ and $w$  are said to be {\bf distinguishable}. 

In this paper, we are using a $b$-bit RAM model of computation and consequently, we assume $n, m, d = O(2^b)$.

\section{Basic method}\label{sec:initial}
Let $G_a$ and $G_b$ be strongly connected permutation digraphs on $n$ vertices, each of degree $d$.
In this section, we show that we can decide whether  $G_a$ and $G_b$ are color-isomorphic by an algorithm running in 
$O(n^{2} + dn \log n)$ time and $O(dn)$ space.  We first prove the following theorem.

\begin{theorem}\label{prop:iso}
Let $G_a$ and $G_b$ be strongly connected permutation digraphs on $n$ vertices and of degree $d$, and 
let  $v_0\in V(G_a)$ and $w_0 \in V(G_b)$. Then there exists a color-isomorphism of $G_a$ onto $G_b$ mapping $v_0$ onto $w_0$ if and only if $v_0$ and $w_0$ are indistinguishable.
\end{theorem}

\begin{proof}
If there exists a color-isomorphism of $G_a$  onto $G_b$  mapping $v_0$ onto $w_0$, then any closed walk is mapped to a closed walk and hence $v_0$ and $w_0$ are indistinguishable.

Let $v_0$ and $w_0$ be indistinguishable vertices of $G_a$  and $G_b$, respectively.
 Construct an isomorphism mapping $v_0$ onto $w_0$ as follows.
 Take a  spanning tree $T_a$ of $G_a$ rooted at $v_0$. For a vertex $v\neq v_0 \in V(T_a)$, let $v_0, \ldots, e^v, v$  be a unique path in $T_a$  from $v_0$ to  $v$, and let  $w_0,\ldots, f^v, w$ be the corresponding path in $G_b$ starting with $w_0$. For each vertex $v\neq v_0  \in V(T_a)$,  map the arc $e^v$ to $f^v$ and the vertex  $v$ to  $w$.
The subdigraph $T_b$ of $G_b$ with the set  $V(G_b)$ of vertices and with the set $\{f^v \in A(G_b)\,|\, v\in V(T_a), v\neq v_0  \}$ of arcs  must be a spanning tree of $G_b$ rooted at $w_0$, for otherwise there exists an open walk starting with $v_0$ such that the corresponding walk  starting with $w_0$ is closed. But this is not possible since $v_0$ and $w_0$ are  indistinguishable. It remains to define the mapping of the cotree arcs.  Let $(u, a_k)$ be a fixed cotree arc of $T_a$, and suppose  that $u$ is mapped to $u'$. Then map   $(u, a_k)$ to the corresponding arc $(u', b_k)$ in $G_b$. If $a_k(u)$ has already been mapped to some vertex different from $b_k(u')$, then there exists a closed walk starting at  $v_0$ such that the corresponding walk starting with $w_0$ is open. But this is not possible. 
Repeating the process at the remaining  cotree arcs yields the desired isomorphism.
\end{proof}

  From Theorem~\ref{prop:iso} we derive the necessary and sufficient condition for $G_a$ and $G_b$ to be color-isomorphic. 
  Since for an arbitrary vertex $v_0 \in V(G_a)$ there are  $n$ possible candidates in $G_b$ that might be indistinguishable from $v_0$,  we need to check, for each vertex $w_0\in V(G_b)$,  whether  $v_0$ and $w_0$ are indistinguishable. 
Algorithm {\sc Indistinguishable}, based on Theorem~\ref{prop:iso}  and breadth-first search, does precisely that -- with an addition of returning the empty word whenever $v_0$ and $w_0$ are indistinguishable, and a distinguishing word for $v$ and $w$ otherwise.

\begin{algorithm}[h!]
\DontPrintSemicolon
\SetNlSty{<text>}{}{:}
\NoCaptionOfAlgo
\SetKw{Comment}{}{}
\SetKwComment{Comment}{ }{}
\SetCommentSty{textrm}
\SetArgSty{text}
\KwIn{Strongly connected permutation digraphs $G_a$ and $G_b$ each of degree $d$ with $n$ vertices,  $v_0\in V(G_a)$ and $w_0\in V(G_b)$.}
\KwOut{The empty word $\epsilon$, if $v_0$ and $w_0$ are indistinguishable, a distinguishing word for $v_0$, and $w_0$ otherwise.}
Construct a spanning tree $T_a$ of $G_a$;\;
$V(T_b) \coloneqq \{w_0\}$; $A(T_b) \coloneqq \emptyset$;\tcp*{\sl a subdigraph of $G_b$}
\lFor(\hfill // {\sl an empty array}){$v \in V(G_a)$}{$\mathcal{D}[v] \coloneqq \text{null}$;}
$\mathcal{D}[v_0] \coloneqq w_0$;\;
$Q \coloneqq \emptyset$;\tcp*{\sl an empty first-in first-out queue}
$Q.\textsc{Enqueue}(v_0)$;\;
\While{$Q \neq \emptyset$}{
$u \coloneqq Q.\textsc{Dequeue}()$;\;
\For{ $e \in A(T_a)$ with $u$ and $v$ as  its end-vertices}{
Let $\mathcal{D}[u], f, v' $ be the path in $G_b$ corresponding to $u, e, v$;\;
\eIf(\hfill // {\sl expand $T_b$} ){$v' \notin V(T_b)$}{
$V(T_b) \coloneqq V(T_b)  \cup \{v'\}$; $A(T_b) \coloneqq A(T_b) \cup \{f\}$; $Q.\textsc{Enqueue}(v)$;  $\mathcal{D}[v] \coloneqq v'$;\;
}
(\hfill // {\sl $v_0$ and $w_0$ are distinguishable} ){
Let $W$ be the closed walk in $G_b$ obtained by following the path in the tree $T_b$ from $w_0$ to $\textrm{ini}(f)$, the edge $f$, and the path in the tree $T_b$ from $\textrm{ter}(f)$ to $w_0$;\; 
\Return $\kappa(W)$;\tcp*{\sl  the word  defined by $W$}
}
}
}
\nonl ~\mbox{}\; \tcp*{\sl $T_B$ is a spanning tree;  need to check if  $\mathcal{D}$ extends to an isomorphism}
\For{ $e\in A(G_a)\setminus A(T_a)$}{
Let $\mathcal{D}[\textrm{ini}(e)], f, v'$ be the path in $G_b$  corresponding to $\textrm{ini}(e), e, \textrm{ter}(e)$;\;
\If(\hfill // {\sl $v_0$ and $w_0$ are distinguishable}){$v' \neq \mathcal{D}[\textrm{ter}(e)]$}
{
Let $W$ be the closed walk in $G_b$ obtained by following the path in $T_b$ from $w_0$ to $\textrm{ini}(f)$, the edge $f$, and the path in $T_b$ from $\textrm{ter}(f)$ to $w_0$;\; 
\Return $\kappa(W)$;\tcp*{\sl  the word  defined by $W$}
}
}

\Return  $\epsilon$;\;

\caption{{\bf Algorithm } {\sc Indistinguishable}\,($G_a, G_b, v_0, w_0$)}
\end{algorithm}

\begin{proposition}\label{prop:Indistinguishable}
Let $G_a$ and $G_b$ be  strongly connected permutation digraphs on $n$ vertices, each of degree $d$, and let $v_0\in V(G_a)$, $w_0\in V(G_b)$. The algorithm {\sc Indistinguishable}$(G_a, G_b, v_0, w_0)$ correctly tests  in  $O(dn)$ time and  $O(dn)$ space whether  $v_0$ and $w_0$ are indistinguishable.
\end{proposition}
\begin{proof}
The correctness of the algorithm follows directly from the remarks above and Theorem~\ref{prop:iso}. 

The closed walk $W$ in lines 14 or 19 is constructed at most once, and this can be done  in time $O(n)$ using $O(n)$ space. Finally, since the algorithm is based on breadth-first search its total time is  $O(dn)$, while the space used is also $O(dn)$. 
\end{proof}

 If  $v_0$ are $w_0$ indistinguishable we have found a color-isomorphism. Otherwise, 
we get a distinguishable word $\kappa$ such that the vertex invariant $I^\kappa$ induces non-trivial partitions 
$$\mathcal{P}_{I^\kappa}(G_a) = \{O_\kappa(V(G_a)), C_\kappa(V(G_a))\},  \quad \mathcal{P}_{I^\kappa}(G_b)  = \{O_\kappa(V(G_b)),C_\kappa(V(G_b))\}.$$ 
If   $|O_\kappa(V(G_a))| \neq |O_\kappa(V(G_b))|$, then the partitions $\mathcal{P}_{I^\kappa}(G_a)$ and $\mathcal{P}_{I^\kappa}(G_b)$ are not compatible. Hence   $G_a$ and $G_b$ are not color-isomorphic by (i) of  Lemma~\ref{compatible}. Consequently,  no indistinguishable vertices exist. 

Otherwise, by (ii) of  Lemma~\ref{compatible} we can search for indistinguishable vertices either between $O_\kappa(V(G_a))$ and $O_\kappa(V(G_b))$ or between $C_\kappa(V(G_a))$ and $C_\kappa(V(G_b))$. 
We re-apply algorithm  {\sc Indistinguishable} 
on  the smaller cells. This reduces  the number of possible candidates   from $n$  to at most $n/2$. 

Algorithm {\sc ColorIsomorphic} formally describes this process.   In the $i$-th iteration, it uses the function {\sc Partition} to split the sets $V_i(G_a)$ and $V_i(G_b)$  with respect  to the word $\kappa_{i}$. This is done by traversing the walk $W(\kappa_{i}, v)$ in
$G$, where  either $G$ is $G_a$ and $v \in V_i(G_a)$, or $G$ is $G_b$ and  $v \in V_i(G_b)$.


\begin{algorithm}[h!]
\SetNlSty{<text>}{}{:}
\NoCaptionOfAlgo
\SetKwComment{Comment}{}{}
\SetCommentSty{<text>}
\SetArgSty{text}
\KwIn{Strongly connected permutation digraphs $G_a$ and $G_b$ each of degree $d$ with $n$ vertices. }
\KwOut{{\sc true}, if $G_a$ and $G_b$ are color-isomorphic, {\sc false} otherwise.}
$V_1(G_a) \coloneqq V(G_a)$,  $V_1(G_b) \coloneqq V(G_b)$,  $i \coloneqq 1$\;
\Repeat{$\kappa_{i-1} = \epsilon$ or $|V_{i}(G_a)| \neq |V_{i}(G_b)|$}{
Select  some $v_i$ from $V_i(G_a)$ and some $w_i$ from $V_i(G_b)$\;
$\kappa_{i} \coloneqq $ {\sc Indistinguishable}\,($G_a, G_b, v_{i}, w_{i}$)\;
\If{$\kappa_{i} \neq \epsilon$}{
$C_{\kappa_{i}}(V_i(G_a))$, $O_{\kappa_{i}}(V_i(G_a)) \coloneqq $  {\sc Partition}\,($V_i(G_a), \kappa_{i}$)\;
$C_{\kappa_{i}}(V_i(G_b))$, $O_{\kappa_{i}}(V_i(G_b)) \coloneqq $  {\sc Partition}\,($V_i(G_b), \kappa_{i}$)\;
\eIf{$|C_{\kappa_{i}}(V_i(G_a))| \leq |O_{\kappa_{i}}(V_i(G_a))|$}{
$V_{i+1}(G_a)  \coloneqq C_{\kappa_{i}}(V_i(G_a)),
V_{i+1}(G_b)  \coloneqq C_{\kappa_{i}}(V_i(G_b))$\;
}
{
$V_{i+1}(G_a)  \coloneqq O_{\kappa_{i}}(V_i(G_a)),
V_{i+1}(G_b)  \coloneqq O_{\kappa_{i}}(V_i(G_b))$\;

}
}
Increase $i$ to $i+1$\;
}
\eIf{$\kappa_i = \epsilon$}{
\Return {\sc true}\;
}
{
\Return {\sc false}\;
}

\caption{{\bf Algorithm } {\sc  ColorIsomorphic}($G_a, G_b$)}
\end{algorithm}

\begin{theorem}
Let  $G_a$ and $G_b$ be strongly connected permutation digraphs on $n$ vertices,  each of degree $d$. The algorithm {\sc  ColorIsomorphic}$(G_a, G_b)$ correctly tests  in $O(n^{2} + dn \log n)$  time  and $O(dn)$ space whether $G_a$ and $G_b$ are color-isomorphic.    
\end{theorem}

\begin{proof}

We first show the following loop invariant.

 At the start of each iteration of the repeat loop at lines 2-13 the following holds:
\begin{enumerate}
\item[(i)]  $1 \leq |V_i(G_a)|  \leq {n} /\, {2^{i-1}}$, and
\item[(ii)] $V_i(G_a)$ and $V_i(G_b)$ are cells of the partitions of $V(G_a)$ and $V(G_b)$, respectively, where both partitions are induced by a vertex invariant $\mathcal{I}_{i}$ such that the vertices in $V_i(G_a)$ and  the vertices in $V_i(G_b)$ have the same labels.
\end{enumerate}

Prior to the first iteration, we have $i=1$, $|V_1(G_a)|  = n$. Thus, property (i) of the invariant holds.    By taking $\mathcal{I}_1 = I^{\epsilon}$, the partitions $\mathcal{P}_{\mathcal{I}_1}(G_a)$ and $\mathcal{P}_{\mathcal{I}_1}(G_b)$  are trivial, and so $V_1(G_a) \in \mathcal{P}_{\mathcal{I}_1}(G_a), V_1(G_b) \in \mathcal{P}_{\mathcal{I}_1}(G_b)$, and the vertices in $V_i(G_a)$ and the vertices in $V_i(G_b)$ have the same labels. Thus,  property (ii)  holds as well.

To see that each iteration maintains the loop invariant, we inductively assume that before the  iteration  $i$,  properties (i) and (ii) hold and that  at the end of  iteration $i$ the condition in line 13 is false. Consequently,  $\kappa_{i} \neq \epsilon$ and $|V_{i+1}(G_a)| = |V_{i+1}(G_b)|$. 
Without loss of generality we may assume that $V_{i+1}(G_a)  = C_{\kappa_{i}}(V_i(G_a))$ and $V_{i+1}(G_b)  = C_{\kappa_{i}}(V_i(G_b))$.

First, since $\kappa_{i}$  is a distinguishable word for some $v_i\in V_i(G_a)$ and some  $w_i\in V_i(G_b)$, at least one of the cells$C_{\kappa_{i}}(V_i(G_a))$ and $C_{\kappa_{i}}(V_i(G_b))$ is nonempty.  But then, since $|C_{\kappa_{i}}(V_i(G_a))| = | C_{\kappa_{i}}(V_i(G_b))|$,  both  $V_{i+1}(G_a)$ and  $V_{i+1}(G_b)$ are nonempty.  Moreover, since  $|V_i(G_a)|  \leq {n}/\,{2^{i-1}}$, and since 
  $ |V_{i+1}(G_a)| \leq |V_{i}(G_a)|/\,2$, it follows that $ |V_{i+1}(G_a)| \leq {n}/\,{2^{i}}$, and the first part of the invariant holds.

To prove (ii), let $\mathcal{I}_{i+1}\colon  \mathcal{G} \to \{{\tt C}, {\tt O}\}^{i+1}$ be a function defined by $
\mathcal{I}_{i+1}(G,v) = \mathcal{I}_i(G,v)\cdot I^{\kappa_{i}}(G,v). 
$
Note that because $\mathcal{I}_i$ and $I^{\kappa_{i}}$ are vertex invariants,  so is $\mathcal{I}_{i+1}$. Since  $V_i(G_a) \in \mathcal{P}_{\mathcal{I}_i}(G_a)$ and   $V_i(G_b)\in \mathcal{P}_{\mathcal{I}_i}(G_b)$, and since $V_{i+1}(G_a)$ and   $V_{i+1}(G_b)$ are nonempty,  it follows that
$$
V_{i+1}(G_a) = C_{\kappa_{i}}(V_i(G_a)) \in \mathcal{P}_{\mathcal{I}_{i+1}}(G_a) \textrm{ and }
 V_{i+1}(G_b) = C_{\kappa_{i}}(V_i(G_b)) \in \mathcal{P}_{\mathcal{I}_{i+1}}(G_b).
$$
Furthermore,   $\mathcal{I}_{i+1}(G_a,v) = \mathcal{I}_i(G_a,v)\cdot {\tt C}$ for all $v\in V_{i+1}(G_a)$, while $\mathcal{I}_{i+1}(G_b,w) = \mathcal{I}_i(G_b,w)\cdot {\tt C}$ for all $w\in V_{i+1}(G_b)$. But $\mathcal{I}_i(G_a, v) = \mathcal{I}_i(G_b, w)$ for all $v\in V_i(G_a)$ and  $w\in V_i(G_b)$, and so $\mathcal{I}_{i+1}(G_a,v) = \mathcal{I}_{i+1}(G_b,w)$ for all $v\in V_{i+1}(G_a)$ and  $w\in V_{i+1}(G_b)$. Thus, the  second part of the invariant  holds.

Next,  by property (i), it follows that the algorithm  terminates  in at most $\lfloor\log n \rfloor +1$ iterations. 
Furthermore, we show that on termination  the algorithm returns {\sc true} if and only if the digraphs are color-isomorphic.
We first show that, if $G_a$ and $G_b$ are color-isomorphic, the algorithm returns {\sc true}. Suppose for the purpose of deriving a contradiction that at the end of some iteration, say $i$, a word $\kappa_{i}$ is  nonempty  and $|V_{i+1}(G_a)| \neq |V_{i+1}(G_b)|$. Then, by  property (ii), 
$V_i(G_a)\in \mathcal{P}_{\mathcal{I}_{i}}(G_a)$ and $V_i(G_b) \in \mathcal{P}_{\mathcal{I}_{i}}(G_b)$, where $\mathcal{I}_{i}$ is a  vertex invariant such that vertices in $V_i(G_a)$ and vertices in $V_i(G_b)$ have the same labels. Consider the vertex invariant $\mathcal{I}_{i+1}$ as defined above. Because $|V_{i+1}(G_a)| \neq |V_{i+1}(G_b)|$, it follows that the partitions  $\mathcal{P}_{\mathcal{I}_{i+1}}(G_a)$ and $\mathcal{P}_{\mathcal{I}_{i+1}}(G_b)$ are not compatible.
 However, since $G_a$ and $G_b$ are assumed to be color-isomorphic, we have by (i) of  Lemma~\ref{compatible} that  $\mathcal{P}_{\mathcal{I}_{i+1}}(G_a)$ and $\mathcal{P}_{\mathcal{I}_{i+1}}(G_b)$ must be  compatible,  which is a contradiction. Conversely, if the algorithm returns {\sc true}, then  at the end of some iteration, say $i$,  a word $\kappa_i$ must be  empty, so there exists a vertex in $G_a$ which is indistinguishable from some vertex in $G_b$. By Theorem~\ref{prop:iso}, $G_a$ and $G_b$ are color-isomorphic. This completes the proof of correctness of the algorithm.

Finally, we compute the  complexity of the algorithm. The initialization in line 1  requires  $O(n)$  time.  From the above analysis it follows that the repeat loop  is executed  at most $\lfloor\log n \rfloor +1$ times. Lines 3 and 12  require $O(1)$  time, while line 4 requires $O(dn)$ time by Proposition~\ref{prop:Indistinguishable}. Let $T_{\sc P}(i)$ denote the time of the  {\sc Partition} to split up $V_i(G_a)$ and $V_i(G_b)$ in lines 6-7 with respect to $\kappa_i$.  Then, we can express the   total time of the algorithm as being bounded from above by
\begin{align}\label{time_naive}
\sum_{i=1}^{\lfloor\log n \rfloor +1} O(dn +T_{\sc P}(i)) & = O\bigg(\sum_{i=1}^{\lfloor\log n \rfloor +1}  (dn + T_{\sc P}(i)) \bigg)\\
& = O\bigg(dn \log n  +  \sum_{i=0}^{\lfloor\log n \rfloor +1}  T_{\sc P}(i)\bigg)\nonumber.
 \end{align}
Since  {\sc Partition} requires time proportional to the product of the length of a word times the number of vertices in a cell we are splitting, and since, at each iteration $i$,   we have $|\kappa_i| = O(n)$   and  $|V_i(G_a)| = |V_i(G_b)| \leq  n /\,{2^{i-1}}$, it holds that  $T_{\sc P}(i) = O(n^2 /\,{2^{i-1}})$. It follows that
\begin{align*}
\sum_{i=0}^{\lfloor\log n \rfloor +1}  T_{\sc P}(i)  & = \sum_{i=0}^{\lfloor\log n \rfloor +1}  O(n^2 /\,{2^{i-1}}) \\
 & =  O\bigg( \sum_{i=0}^{\lfloor\log n \rfloor +1}  n^2 /\,{2^{i-1}}\bigg) = O\bigg( n^2 \sum_{i=0}^{\lfloor\log n \rfloor +1}  \frac{1}{2^{i}} \bigg).
\end{align*}
Since $\sum_{i=0}^{\infty}  1/\,{2^{i}} = 2$, we can bound the running time of the algorithm as
$$
O\bigg(dn \log n  + n^2 \sum_{i=0}^{\lfloor\log n \rfloor +1}  \frac{1}{2^{i}}\bigg) =   O\bigg(dn \log n  + n^2 \sum_{i=0}^{\infty}  \frac{1}{2^{i}}\bigg) = O(n^{2} + dn \log n).
$$
The space complexity of the algorithm is clearly  $O(dn)$. This completes the proof.
\end{proof}


\section{Towards a subquadratic time algorithm}\label{sec:shorter}

The running time of  the algorithm {\sc ColorIsomorphic} is dominated by  the running time of the function {\sc Partition}; in turn, this depends on the length of a distinguishing word $\kappa$. Since   the output sets $O_\kappa(V(G_a))$  and $C_\kappa(V(G_a))$ of  {\sc Partition}$(V(G_a), \kappa)$ are the support set and the fixed points set of the permutation $a_{\kappa}$, respectively, we consider the following a bit more general problem:

\begin{quote}\it
Given a finite  sequence  $a_1, a_2,\ldots, a_d$ of $d$ permutations in $S_n$ and a  word  
$\kappa = k_1k_2\cdots k_{m}$ over $[d]$, can we evaluate   the product $a_{\kappa} = a_{k_1}a_{k_2}\cdots a_{k_{m}}$ on $s$ points in  time  $o(sm)$?
\end{quote}

An especially interesting case for us is when   $m, s = \Theta(n)$.  In this case,  we show that, at a given $d$,  the product $a_{\kappa}$ can be evaluated 
in time  $o(n^2)$. This reduces the running time of  the algorithm {\sc ColorIsomorphic}   to  subquadratic in $n$.

Our approach is similar to the  fast computation of large positive integer powers $\beta^m$ by repeated squaring. 
For technical reasons, we describe a procedure of how  to  make $\kappa$ always of even length. Namely,  if $|\kappa|$ is odd we expand it by a single character $*\notin \Sigma$ making its length even. Moreover, defining $a_{*}$ to be the identity permutation of $[n]$, the new word  $\kappa$ defines the same product $a_{k_1}a_{k_2}\cdots a_{k_{m}}$. 

Let $\mathcal{S}_0 = \{a_1,a_2,\ldots, a_d\}$,    $\kappa_0 = \kappa$, $k_{0, j} = k_j$, $j\in [m]$, and  $\Sigma_0 = [d]$. 
In case  $|\kappa_0|$ is odd, we use the above technical procedure making its length even.
In the  next step we scan through the word $\kappa_0$, replacing a pair $k_{0, 2j-1}k_{0, 2j}$ by $k_{1, j}$, $j\in  [\lc m/2\rc]$.  The obtained word $\kappa_1 = k_{1 ,1} k_{1 ,2} \cdots k_{1 ,\lc m/2\rc}$ is actually built over $\Sigma_1 \cup \{k_{1, \lc m/2\rc}\}$, where $\Sigma_1=  \{ ij \,|\,  i,j\in \Sigma_0\}$.  Clearly,  $|\kappa_1| = \frac{1}{2} |\kappa_0| $ and $a_{\kappa_0} = a_{\kappa_1}$. So assuming that  the set $\mathcal{S}_1 = \{ a_{ij}  = a_ia_j \,|\, a_i, a_j\in \mathcal{S}_0\}$  of all products of pairs of permutations in
 $\mathcal{S}_0$ is precomputed and the permutation $a_{k_{1, \lc m/2\rc}}$  is known,  the  time of straightforward evaluation of $a_{\kappa}$   is reduced  to half.

 The above reduction step can be repeated. If before the $t$-th iteration the length   $|\kappa_{t-1}|$ is odd, we apply the above technical procedure making $|\kappa_{t-1}|$ even. 
Then, after the $t$-th iteration we have the set  $\Sigma_t = \{ ij \,|\,  i,j\in \Sigma_{t-1}\}$ and
the word 
$$\kappa_t = k_{t, 1}k_{t, 2}\cdots k_{t, \ell},$$
where
  $k_{t, j}= k_{t-1, 2j-1} k_{t-1, 2j}$ for $j\in 
 [\ell]$, and $\kappa_t$ is over the alphabet  $\Sigma_t \cup \{k_{t, \ell} \}$.  We leave to the reader to check that $\ell = \lc m/2^t\rc$. Again, with  
$\mathcal{S}_t = \{ a_{ij}  = a_ia_j \,|\,  a_i, a_j \in \mathcal{S}_{t-1}\}$ and the permutation 
$a_{k_{t, \ell}}$ being precomputed,  straightforward evaluation of the product $a_{\kappa_t}$ takes $O(n \cdot m/2^t)$.

At this point, we truncate the iteration similarly as it is the
recursion in \cite{BrodnikMunro}.
Namely, since  $|\mathcal{S}_t | =  d^{2^t}$, the construction of $\mathcal{S}_t$ from $\mathcal{S}_{t-1}$ takes $O(n \cdot d^{2^t})$.  Consequently, we  iterate the described process for $\nu$ times until
 \begin{equation}\label{ineq}
 d^{\,2^{\nu}}\leq \frac{m}{2^{\nu}}.
\end{equation}
%
Observe that  inequality (\ref{ineq})  has no closed form solution for  ${\nu}$. However, by rewriting it  as $2^{\nu} \leq \log_d m  - {\nu} \log_d 2 $, we see that $2^{\nu} \leq \log_d m $. 
Finally, we define  $\nu$ to be the integer satisfying
%
 \begin{equation}\label{nu}
\frac{1}{4} \log_d m  < 2^{\nu} \leq  \frac{1}{2} \log_d m.
\end{equation}
Therefore, after ${\nu}$ iterations, the length of the word $\kappa_{\nu}$  becomes  bounded from above by
$$|\kappa_{\nu}| = \frac{m}{2^{\nu}}< \frac{4m }{ \log_d m } = 4\cdot \frac{m \log d}{ \log m} = O(  m \log d / \log m).$$

Based on this analysis, we give a formal description of this reduction in the algorithm {\sc WordReduction}. This algorithm  guarantees that the size of the word is $O(  m \log d / \log m)$, and simultaneously increases the set of permutations.
\begin{algorithm}[h!]
\SetNlSty{<text>}{}{:}
\NoCaptionOfAlgo
\SetKwComment{Comment}{}{}
\SetCommentSty{<text>}
\SetArgSty{text}
\KwIn{A set $\mathcal{S}_0 = \{a_1, a_2,\ldots, a_d\}$ of permutations of $[n]$, a word $\kappa_0 = k_{0, 1}k_{0, 2}\cdots k_{0, m}$ over $ [d]$.}
\KwOut{a set of permutations and a word of length $O(m \log d / \log m)$.}
\eIf{$m < d^{\,4 }$}
{\Return {$\kappa_0$, $\mathcal{S}_0$}\; }
{Let $\nu$ be a positive integer satisfying $1 /\,4 \log_d m  < 2^{\nu} \leq  1 /\,2 \log_d m$\;
\For{$ t = 1$ \KwTo $\nu$}
{
$\ell \coloneqq |\kappa_{t-1}|$\;
\If{$\ell$ is odd}{
$\kappa_{t-1} \coloneqq \kappa_{t-1} *$\;
}
$\kappa_{t} \coloneqq k_{t, 1}k_{t, 2}\cdots k_{t, \lc \ell/2\rc}$ where $k_{t, j} \coloneqq k_{t-1, 2j-1} k_{t-1, 2j}$, $j\in [\lc \ell/2\rc]$\;
$\mathcal{S}_t \coloneqq \{ a_{ij} \,|\, a_{ij}  = a_ia_j, a_i, a_j\in \mathcal{S}_{t-1}\}$\;
$a_{k_{t, \lc \ell/2\rc}} \coloneqq a_{k_{t-1, 2\lc \ell/2\rc- 1 }} a_{k_{t-1, 2\lc \ell/2\rc}}$\;
}
{\Return {$\mathcal{S}_{\nu} \cup \{a_{\kappa_{\nu},|\kappa_{\nu}|} \}$, $\kappa_{\nu}$}\; }
}

\caption{{\bf Algorithm } {\sc  WordReduction}($\mathcal{S}_{0}$, $\kappa_0$)}
\end{algorithm}

\begin{lemma}\label{reduction:ana}
Let  $\mathcal{S}_{0} = \{a_1, a_2, \ldots, a_d\}$ be a set of permutations of $[n]$, and let  $\kappa_0 = k_1k_2\cdots k_m$ be a word  over $[d]$. The algorithm  {\sc  WordReduction}$(\mathcal{S}_{0}$, $\kappa_0)$ takes $O(m + n\sqrt{m})$ time and $O(m + n\sqrt{m})$ space.
\end{lemma}
\begin{proof}
We may assume that $m \geq d^{\,4}$ as otherwise we are done. 
Lines 6-8 take $O(1)$ time. Since $|\kappa_{t-1}| = \lc m/\, 2^{t- 1} \rc$, lines 8 and 9 take
$O( m/\, 2^{t- 1} )$ time.
Obviously,  the  construction of $\mathcal{S}_t$ from $\mathcal{S}_{t-1}$ in line  10 takes  $O(nd^{\,2^t})$ time, while line 11 takes $O(n)$ time.
So, we can express the   total running time of the algorithm as being bounded from above by
$$
\sum_{t=1}^{\nu} O(m/2^{t-1} + nd^{\,2^t})  = O\bigg(m + n \sum_{t=1}^{\nu} d^{\,2^t}  \bigg). 
$$
The last summation can be estimated as
$$ \sum_{t=1}^{\nu} d^{\,2^t}  \leq  \sum_{t=0}^{\infty} \frac{d^{\,2^{\nu}}}{2^t}  = 2 d^{\,2^{\nu}}   < 2 d^{ \,\log_d \sqrt{m}}  = 2 \sqrt{m}.$$
Thus, the algorithm takes  $O(m + n \sqrt{m})$ time.
Moreover, the space complexity is proportional  to the sum of the length of the input word and the size  of the set
$\mathcal{S}_{\nu}$. However,   $\mathcal{S}_{\nu}$ has  $d^{\,2^{\nu}} = O(\sqrt{m})$ permutations of length $n$, and the result follows.
\end{proof}

Once the length of the word is guaranteed to be $O(m \log d / \log m)$, we evaluate its corresponding product  on $s$ points in a straightforward manner.

\begin{lemma}\label{lemma:reduce}
Let $a_1, a_2, \ldots, a_d$ be permutations of  $[n]$ and let $k_1k_2\cdots k_m$ be a word  over  $[d]$. Then we can evaluate the product $a_{\kappa_1}a_{\kappa_2}\ldots a_{\kappa_m}$ on $s$ points 
 in $O(m + n\sqrt{m} + s\min\{m, m  \log d / \log m\})$ time and $O(m + n\sqrt{m})$ space.
\end{lemma} 
\begin{proof}
The evaluation is done in two phases. First, by Lemma~\ref{reduction:ana} we reduce the length of the word  to  $O(m \log d / \log m)$ in  $O(m + n\sqrt{m})$ time and $O(m + n\sqrt{m})$ space. And second,  we evaluate the obtained word  in  $O(s\min\{m, m  \log d / \log m\})$ time.
\end{proof}

\begin{remark}\label{rem:partial}
Lemma~\ref{lemma:reduce} partially answers the question we posed at the beginning of this section. The evaluation  can be done in time $o(sm)$ whenever $s= \Theta(n)$ and $\log d= o(\log m)$.
\end{remark}

We are now ready to prove  Theorem~\ref{thm:main}.

\begin{proof}[Proof of Theorem~\ref{thm:main}]
Given $d$-tuples $a$ and $b$ of permutations  that generate transitive subgroups of the symmetric group $S_n$, the tuple $a$ is  simultaneously  conjugate to tuple $b$ if and only if the permutation digraph $G_a$
 is color-isomorphic  to the permutation digraph $G_b$. To this end, we modify the basic method in  the algorithm  {\sc ColorIsomorphic}  by speeding-up the bottleneck which is {\sc Partition}.

Instead of using  {\sc Partition}  to partition the sets $V_i(G_a)$ and $V_i(G_b)$ of sizes $s= O(n / 2^{i-1})$ with respect to a distinguishing word $\kappa_{i+1}$ of length $m= O(n)$, we can compute, by  Lemma~\ref{lemma:reduce}, the cells $C_{\kappa_{i+1}}(V_i(G_a))$, $O_{\kappa_{i+1}}(V_i(G_a))$, $C_{\kappa_{i+1}}(V_i(G_b))$, and $O_{\kappa_{i+1}}(V_i(G_b))$ in  
\begin{equation}\label{eq:time-bound}
O\bigg(n + n^{3/2} + n^2  \frac{\log d}{2^{i-1} \log n}\bigg) 
\end{equation}
 time  and  $O(n^{3/2} )$ space. Then, by replacing $T_{\sc P}(i)$ in equation (\ref{time_naive}) with (\ref{eq:time-bound}) 
 we can bound the running time of the algorithm to
$$O\left(dn \log n  +   n^{3/2} \log n + n^2  \frac{\log d}{ \log n}\right) = O\bigg( dn\log n + n^2  \frac{\log d}{ \log n}\bigg).$$

In addition to $O(dn)$ space required by  the basic method, we need an extra   $O(n^{3/2} )$ space to compute the cells
 and the result follows.
\end{proof}

\section{Strongly subquadratic time algorithm}\label{sec:storngly}
If we could guarantee that a spanning tree $T_a$ ($a= (a_1, a_2, \ldots, a_d)$) in line 1 of  {\sc Indistinguishable} has many arcs of the same color, then any distinguishing word arising from such a tree would contain many occurrences of some character. Consequently, the  evaluation of the corresponding permutation can be done even more efficiently than in the previous section.

To achieve this,  let  $\lambda_i$ 
 be the number of cycles in  the cycle decomposition of $a_i$, and let us define $\lambda = \min_{i\in [d]} \lambda_i$. 
We present a construction of $T_a$ which for $\lambda = O(n^{\epsilon})$, where $0 \leq \epsilon < 1$, results in  a  strongly subquadratic  running time in $n$ of the algorithm {\sc ColorIsomorphic}.

Let  $a_j$ be a permutation with  $\lambda_j = \lambda$ cycles, and let $C_1, C_2,\ldots, C_{\lambda}$ be their corresponding subdigraphs in $G_a$. We construct $T_a$ in two steps: 
we first construct a weakly connected subdigraph $S_a$ of $G_a$, and then we find the required spanning tree $T_a$ as a  spanning tree in $S_a$. More precisely, $V(S_a) = V(G_a)$ while $A(S_a)$ consists of  all arcs $A(C_i)$, where $i \in [\lambda]$, and 
 $\lambda-1$ arcs  $e_1, e_2,\ldots, e_{\lambda-1}$ (of colors different from $j$)  such that  $S_a$ is weakly connected.
This can be achieved in $O(dn)$ time by slightly modifying a breadth first search on $G_a$ such that, as soon as a vertex in $C_i$ is visited, all other vertices in $C_i$ are marked as visited as well.  With $S_a$ in hand, we construct a subdigraph $T_a$  by deleting an arbitrary arc from each  $A(C_i)$. The above construction  is summarized in the following lemma.

\begin{lemma}\label{tree-con}
With the notation and assumptions of this section, we can compute a spanning tree  with  $\lambda - 1$ arcs of colors different from $j$ in $O(dn)$ time and $O(dn)$ space.
\end{lemma}

Any distinguishing word $\kappa$  over $\Sigma = \{k^{\alpha} \,|\, k \in [d], \alpha = \pm 1\}$ arising  from  a spanning tree in Lemma~\ref{tree-con} contains at  most   $\lambda - 1$ characters different from $j^{+1}$ and $j^{-1}$, and can, therefore, be written without loss of generality as
\begin{equation}\label{eq:word}
\kappa = (j^{\alpha_1})^{p_1}\kappa_1(j^{\alpha_2})^{p_2}\kappa_2 (j^{\alpha_3})^{p_3}\cdots  \kappa_{m- 1}(j^{\alpha_m})^{p_{m}}, 
\end{equation}
where each $0\leq p_i < n$, $\alpha_i \in \{-1, +1\}$, $\kappa_i  \in  \Sigma \setminus \{j^{+1}, j^{-1}\}$, and $m \leq \lambda$. Efficient evaluation of  the powers $a_j^{p_i}$ and $(a_j^{-1})^{p_i}$  using  the idea of fast computation of large positive integer powers mentioned in the previous section, results in efficient  evaluation of the  product  corresponding to equality $(\ref{eq:word})$.

\begin{lemma}\label{lemma:eval}
Let $a_1, a_2, \ldots, a_d$ be permutations of  $[n]$, and let  $
j^{p_1}\kappa_1j^{p_2}\kappa_2\cdots  j^{p_{m}}
$ be a word over $[d]$,
where each $0\leq p_i < n$ and $\kappa_i   \in [d]\setminus \{j\}$. Given the powers $a_j^2, a_j^{2^2}, \ldots, a_j^{2^{\lfloor\textrm{log}(n)\rfloor}}$, 
 we  can evaluate  the product $a_{j}^{p_1}a_{\kappa_1}a_{j}^{p_2}a_{\kappa_2}\cdots a_{j}^{p_{m}}$ on $s$ points 
 in $O(s m \log n)$ time.
\end{lemma} 
\begin{proof}
Note that each $p_i$ can be written as $p_i = \sum_{k=0}^{\lfloor \log n\rfloor} c_k{2^k}$, where each $c_k\in \{0,1\}$. For the given  powers
$
a_j^2, a_j^{2^2}, \ldots, a_j^{2^{\lfloor \log n\rfloor}}
$
we can obviously evaluate each  $a_j^{p_i}$
at any point  in $O(\log n)$ time. The result of lemma trivially follows.
\end{proof}

We are now ready to prove  Theorem~\ref{th:strongly}, which implies  strongly subquadratic time in $n$ at a given $d$ 
 as soon as $\lambda = O(n^{\epsilon})$ for some constant $0 \leq \epsilon < 1$.  In particular, for $\lambda = O(1)$ we have the following obvious corollary to  Theorem~\ref{th:strongly}.
 
 \begin{corollary}
If $\lambda = O(1)$, then  the transitive $d$-SCP in the symmetric group $S_n$ can be solved in time $O(dn\log n)$.
\end{corollary}

\begin{proof}[Proof of Theorem~\ref{th:strongly}]


By Lemma~\ref{tree-con} we can compute a spanning tree $T_a$ with  $\lambda - 1$ arcs of colors different from $j$ in $O(dn)$ time and $O(dn)$ space. Taking such a $T_a$  in  line 1 of  {\sc Indistinguishable}, any distinguishing word arising from $T_a$ can be written as in  equality (\ref{eq:word}), which is of length $O(\lambda)$.  We then compute the inverse permutations $a_1^{-1}, a_2^{-1}, \ldots, a_d^{-1}$ and $b_1^{-1}, b_2^{-1}, \ldots, b_d^{-1}$. This takes $O(dn )$ time and $O(dn)$ space. Next, for each $c\in \{a_j, a_j^{-1}, b_j, b_j^{-1}\}$ we compute 
the powers
$$
c_j^2, c_j^{2^2}, \ldots, c_j^{2^{\lfloor\textrm{log}(n)\rfloor}} 
$$
using the standard approach, which requires  $O(n \log n)$ time and $O(n \log n)$ space. 
Then, to partition the sets $V_i(G_a)$ and $V_i(G_b)$ of sizes $s= O(n / 2^{i-1})$ with respect to a distinguishing word $\kappa_{i+1}$ (of length $O(\lambda)$) 
we compute the cells $C_{\kappa_{i+1}}(V_i(G_a))$ and  $O_{\kappa_{i+1}}(V_i(G_a))$, as well as   $C_{\kappa_{i+1}}(V_i(G_b))$ and $O_{\kappa_{i+1}}(V_i(G_b))$,  by the procedure described in  Lemma~\ref{lemma:eval} -- instead of by using  {\sc Partition} in   {\sc ColorIsomorphic}. This can be done  in $O(\lambda n \log n/\,2^{i-1})$ time, which is now, in fact, $T_{\sc P}(i)$ in equality (\ref{time_naive}). Consequently,  taking into account the $O(n \log n)$ time for precomputing the powers  $c_j^{2^k}$
 we can bound the running time of the algorithm by
$$ O(dn \log n  +  \lambda n \log n) + O(n \log n) = O((d + \lambda) n  \log n),$$
as claimed.

In addition to $O(dn)$ space required by  the basic algorithm, we need an extra   $O(n (\log n  +d))$ space for storing the inverse permutations and powers,
which proves the space bound and completes the proof.

\end{proof}

\section{Linear time algorithm}\label{linear}
In this section, we consider a special case when there is a permutation which is an $n$-cycle, that is, $\lambda = 1$. Here we take a completely different approach not based on a distinguishing word. Rather, to each permutation digraph, we assign a special string in such a way that the permutation digraphs are color isomorphic if and only if the corresponding strings are cyclically equivalent. This equivalence is then checked by the linear time Knuth-Morris-Pratt string-matching algorithm.

Let $G_a$ be a permutation digraph of degree $d$ with $n$ vertices so that some permutation, say $a_j$, is an $n$-cycle. Without loss of generality we may assume that the vertices are labeled in  such a way that  $a_j(i) = i \bmod  n + 1$, $i\in [n]$, and we consider $a_j$ as a reference cycle. To $G_a$,  we assign a string over $\{0, 1, \ldots, n\}$ in three steps as follows. First, we label the arcs of $G_a$: each arc $e$ with $c(e) \neq j$ is labeled  by $\sigma^j(e)$ that indicates the number of color-$j$ arcs along the reference cycle $a_j$ from its initial vertex $\textrm{ini}(e)$ to its terminal vertex $\textrm{ter}(e)$; while  arcs  $e$ with $c(e) = j$ are labeled by $\sigma^j(e) = n$. More formally, for each arc $e= (i, a_k)$, 
$$
 \sigma^j(i, a_k)=\left\{
  \begin{array}{@{}ll@{}}
    (a_k(i) - i)  \bmod n, & \text{if}\  k \neq j\\
    n, & \text{otherwise.}
  \end{array}\right.
$$  
Second,  using these arc-labels we label each vertex  $i$ of $G_a$  by  the string
$$\Delta^j_i(G_a)= \sigma^j(i, a_1)\cdot \sigma^j(i, a_2) \cdots \sigma^j(i, a_d),$$
which is a concatenation of labels of all arcs out-going from $i$, ordered by color.
Finally, we encode $G_a$ as a string $\mathcal{S}^j(G_a)= \Delta^j_1(G_a)\cdot \Delta^j_2(G_a) \cdots \Delta^j_n(G_a)$.

\begin{proposition}\label{prop:equiv}
Let $G_a$ and $G_b$ be  permutation digraphs on $n$ vertices each of degree $d$ where, for some $j\in [n]$, the permutations $a_j$ and $b_j$ are $n$-cycles. Then there exists a color-isomorphism of $G_a$ onto $G_b$ if and only if the corresponding strings $\mathcal{S}^j(G_a)$ and $\mathcal{S}^j(G_b)$ are cyclically equivalent.
\end{proposition}
\begin{proof}
Suppose first that there exists a color-isomorphism $f \colon G_a \to G_b$. Let $x= \Delta^j_1(G_a)\cdot \Delta^j_2(G_a) \cdots \Delta^j_{f(1) -1}(G_a)$ and $y= \Delta^j_{f(1)}(G_a)\cdot \Delta^j_{f(1) + 1}(G_a) \cdots \Delta^j_{n}(G_a)$. It is easy to see that $\mathcal{S}^j(G_b) = yx$, and so  $\mathcal{S}^j(G_a)$ and $\mathcal{S}^j(G_b)$ are cyclically equivalent.

Conversely,  suppose that the strings $\mathcal{S}^j(G_a)$ and $\mathcal{S}^j(G_b)$ are cyclically equivalent. By definition, in both strings the character $n$ occurs precisely $n$ times at positions $j, j + n,\ldots, j + (d-1)n.$ So, there must exists $k\in [n]$ such that $x= \Delta^j_1(G_a)\cdot \Delta^j_2(G_a) \cdots \Delta^j_{k -1}(G_a)$,  $y= \Delta^j_{k}(G_a)\cdot \Delta^j_{k + 1}(G_a) \cdots \Delta^j_{n}(G_a)$ and $\mathcal{S}^j(G_b) = yx$. It is straightforward to check that the mapping $f\colon V(G_a) \to V(G_b)$, defined by 
$f(i) =   (i + k - 2) \bmod n + 1$ extends to a color-isomorphism of $G_a$ onto $G_b$.
This completes the proof.
\end{proof}

We are now ready to prove Theorem~\ref{th:linear}.


\begin{proof}[Proof of Theorem~\ref{th:linear}]
First, we can construct the corresponding strings $\mathcal{S}^j(G_a)$ and $\mathcal{S}^j(G_b)$ of length $dn$ in $O(dn)$ time and $O(dn)$ space. Next, by Proposition~\ref{prop:equiv}, testing $G_a$ and $G_b$ for color isomorphism is equivalent to testing 
whether $\mathcal{S}^j(G_a)$ and $\mathcal{S}^j(G_b)$ are cyclically equivalent. In turn, this is equivalent to  asking whether $\mathcal{S}^j(G_b)$ is a substring of the string $\mathcal{S}^j(G_a)\cdot \mathcal{S}^j(G_a)$. Finally, the last problem can be solved  in $O(dn)$ time and $O(dn)$ space using the Knuth-Morris-Pratt string-matching algorithm \cite{KMP77}, which concludes the proof.
\end{proof}

\section{Empirical results}\label{sec:empir}

In empirical evaluation we compared three algorithms. The quadratic one is the algorithm described in \cite{Fontet, Hoff82}, the subqaudratic one is the algorithm {\sc ColorIsomorphic} with {\sc WordReduction} improvement, and the linear one is the algorithm described in Section~\ref{linear}. All algorithms have been implemented in the system {\sc GAP} \cite{gap}. A breadth-first search has been used to construct a spanning tree in line 1 of {\sc Indistinguishable}.

We performed two experiments: in the first experiment we compared the subquadratic algorithm with the quadratic algorithm in the general transitive case, while in the second experiment we compared the linear algorithm with the subquadratic one in a special case when the first permutation in each tuple was an $n$-cycle.


The tests were conducted on isomorphic and non-isomorphic pairs of randomly generated tuples. 
For isomorphic pairs in the first experiment, the first tuple  was obtained by repeatedly selecting   random elements of $S_n$ until a transitive subgroup was generated. The second tuple, isomorphic to the first one, was obtained by conjugation with a random permutation. For isomorphic pairs in the second experiment, the first tuple  was obtained by repeatedly selecting 
 $\log n$ random elements of $S_n$ in  such a way that the  first element was an $n$-cycle. The isomorphic tuple was then generated in the same way as in the first experiment. 
 
For non-isomorphic pairs in the first experiment, we initially generated a random transitive tuple $(a_1,a_2,\ldots, a_d)$  in such a way that $a_1^2 \neq 1$. Next, we randomly chose a permutation $\tau$ such that $\tau a_1^2 \neq a_1^2 \tau$, and constructed the tuples   $(a_1,a_2,\ldots, a_d, a_1^2)$ and  $(\tau^{-1}a_1\tau, \tau^{-1}a_2\tau ,\ldots, \tau^{-1}a_d\tau, a_1^2)$.  It is not hard to check that such pair of tuples is non-isomorphic. Non-isomorphic pairs in the second experiment were generated in the same way as in the first experiment with   $a_1$ being an $n$-cycle and $d= \log n$.

The results of the first experiment for  sizes of $n$ varying  between 10\;000 and 50\;000 are presented in
Table~\ref{tab:1}, while the results of the second experiment for sizes of $n$ varying  between
100\;000 and 500\;000 are shown in Table~\ref{tab:2}.
Runtimes, given in miliseconds, were measured by {\sc GAP}'s  {\sc NanosecondsSinceEpoch} function on a Macbook Pro with 2,9 GHz Intel Core i5   processor under macOS Sierra version 10.12.6.

 Table~\ref{tab:1} shows that  the subquadratic algorithm is a clear winner of the comparison. Moreover, its running time in practice is much faster than our worst-case estimate. The main reason for this is that the running time is dominated by the length of   distinguishing words, and this in turn depends on the depth of the breadth-first search tree constructed  in line 1 of {\sc Indistinguishable}, which, in practice,  is usually $O(\log n)$ with a modest constant \cite{Seress}. 

\begin{table}[ht!]
\footnotesize
\centering
\begin{threeparttable}
\captionof{table}{Experimental results  in miliseconds on random instances  for the general transitive case.}
\label{tab:1}
\begin{tabular}{@{}lrrrcrr@{}}\toprule
  &\phantom{} & \multicolumn{2}{c}{isomorphic pairs} & \phantom{}& \multicolumn{2}{c}{non-isomorphic pairs}  \\
\cmidrule{3-4} \cmidrule{6-7} 
$n$ && subquadratic & quadratic &  & subquadratic &  quadratic   \\ \midrule
10\;000 & & 46 &     10\;578   && 13 & 17\;865\\ 
15\;000 & & 69 &    82     && 10 & 39\;571    \\
20\;000 & & 111 &    18\;398 &&    12    & 67\;958   \\
25\;000 && 147   & 63\;750    && 14    & 136\;404    \\
30\;000 && 283    & 14\;378   && 16    & 196\;014   \\
35\;000 && 294    & 69\;366   && 18   & 265\;281    \\
40\;000 && 378 &    28\;471  &&    24    & 347\;392 \\
45\;000 && 478 &   161\;269 &&    88    & 390\;613 \\
50\;000 && 556 &    105\;583  &&    27    & 438\;985 \\
 \noalign{\smallskip}
\bottomrule
\end{tabular}                   
\end{threeparttable}
\end{table}

The results in Table~\ref{tab:2} show that the linear algorithm outperforms the subquadratic one on isomorphic pairs while on non-isomorphic pairs  the opposite happens. The main reason for this is that, in contrast with the linear  algorithm  the subquadratic algorithm does not necessary scan the entire input to find that given tuples are non-isomorphic.

\begin{table}[ht!]
\footnotesize
\centering
\begin{threeparttable}
\captionof{table}{Experimental results  in miliseconds on random instances for a special case with an $n$-cycle.}
\label{tab:2}
\begin{tabular}{@{}lrrrcrr@{}}\toprule
  &\phantom{} & \multicolumn{2}{c}{isomorphic pairs} & \phantom{}& \multicolumn{2}{c}{non-isomorphic pairs}  \\
\cmidrule{3-4} \cmidrule{6-7} 
$n$ && linear & subquadratic &  & linear &  subquadratic   \\ \midrule
100\;000 & & 2\;079 &     5\;818   && 2\;927 & 106\\ 
150\;000 & & 4\;674 &    11\;581   && 3\;457 & 124    \\
200\;000 & & 4\;157 &    18\;666 &&    4\;888    & 1\;877  \\
250\;000 && 7\;089  & 29\;622   && 6\;424    & 191   \\
300\;000 && 8\;879    & 40\;810  && 7\;805  & 255  \\
350\;000 && 8\;042    & 63\;686   && 9\;382  & 634    \\
400\;000 && 10\;942 &    75\;584  &&    10\;810   & 355 \\
450\;000 && 13\;308 &   93\;833 &&    12\;495   & 634 \\
500\;000 && 12\;697 &    119\;640  &&    14\;585   & 429 \\
 \noalign{\smallskip}
\bottomrule
\end{tabular}                   
\end{threeparttable}
\end{table}

\section{Concluding remarks}\label{sec:conc}

We have shown that the $d$-SCP in $S_n$ can be solved in  $O(n^2 \log d / \log n + dn\log n)$ worst-case time for the case when the two given tuples of permutations  generate transitive groups. On the other hand,  our experimental results on random instances suggest the following conjecture.
\begin{conjecture}
If  the permutations of transitive tuples are chosen uniformly at random,  the expected running time of the algorithm {\sc ColorIsomorphic} with {\sc WordReduction} improvement is nearly-linear in $n$ at given $d$. 
\end{conjecture}
Next, to to the best of our knowledge, there are no known lower bounds for the SCP, except the trivial linear bound $\Omega(dn)$, which leaves open also the gap between the upper and the lower bound.
 
The key approach used in our solution is that, given  a sequence  $a_1, a_2,\ldots, a_d$ of permutations in $S_n$ and a  word  
$ k_1k_2\cdots k_{m}$ over $[d]$,  we evaluate the product $a_{k_1}a_{k_2}\cdots a_{k_{m}}$ on $n$ points in time 
$o(nm)$ whenever $\log d= o(\log m)$. If $d\geq  m$, the straightforward $O(n m)$-time evaluation  is optimal, however, it remains open how big $d$ can be that there still exists an $o(n m)$-time solution. 

Finally, we remark that we have recently developed a worst-case subquadratic algorithm in $n$ at given $d$ also for the general case, i.e.,  when each tuple of permutations  generates an intransitive group \cite{BMP}. 





\section*{Acknowledgment}
The authors would like to thank Ilia Ponomarenko for enlightening discussions, and to the referee for valuable suggestions.


\bibliographystyle{amsplain}

\pagebreak

\appendixtitleon
\appendixtitletocon
\begin{appendices}

\section{}\label{append}


 \begin{algorithm}[h]
\SetNlSty{<text>}{}{:}
\NoCaptionOfAlgo
\SetKwComment{Comment}{}{}
\SetCommentSty{<text>}
\SetArgSty{text}
\KwIn{a permutation digraph $G_a$ with $n$ vertices and of degree $d$ such that all cycles of a given color are of the same length.}
\KwOut{arc labels $L(e)$ of $G_a$.}
Let $C_0,C_1, \ldots, C_{p -1}$ be $t$-cycles  in a cycle decomposition of the permutation $a_1$, and  let  
$
C_i \coloneqq (v_{it +1},v_{it+2}, \ldots, v_{(i+1)t})
$ 
for each  $i = 0,1,\ldots, p-1$\;
\For{$ i = 1$ \KwTo $n$}
{$l(v_i) \coloneqq i-1$\;}
\For{$ i = 1$ \KwTo $n$}
{$L((i, a_1)) \coloneqq \langle 0, 0\rangle$\;}
\For{$ k = 2$ \KwTo $d$}
{
\For{$ i = 1$ \KwTo $n$}
{
Let $e \coloneqq (i, a_k(i))$\;
$r \coloneqq \lfloor l(i)/t \rfloor$\;
$s \coloneqq \lfloor l(a_k(i))/t \rfloor$\;
\eIf{$r=s$}
{$L(e) \coloneqq  \langle 1, (l(a_k(i)) - l(i)) \mod t \rangle$\;}
{\eIf{there is no arc from $C_r$ to $C_s$ with label $\langle 2, 0\rangle $  }
{$L(e) \coloneqq \langle 2, 0\rangle $\;}
{Suppose $(j, a_k(j))$ is the arc from $C_r$ to $C_s$ with label $\langle 2, 0\rangle $\;
$L(e) =\coloneqq \langle 2, (l(a_k(i)) - l(a_k(j)) - (l(i) - l(j)) \mod t \rangle$\;}
}
}
}
\caption{{\bf Algorithm } {\sc  ArcLabeling}($G_a$)}
\end{algorithm}

 \end{appendices}

\end{document}